\documentclass{article}[11pt]
\usepackage[left=1in, right=1in, top=1in, bottom=1in]{geometry}

\usepackage{amsthm}
\usepackage{amssymb}
\usepackage{amsmath}
\usepackage{mathtools}
\usepackage{csquotes}
\usepackage{enumitem}
\usepackage{svg}
\usepackage{hyperref}
\usepackage{xcolor}
\usepackage{fancyhdr}
\usepackage{mdframed}
\usepackage{framed}

\fancypagestyle{plain}{
\fancyhf{}
\fancyfoot[OR]{\thepage}
\fancyfoot[EL]{\thepage}
 
}

\hypersetup{
    colorlinks,
    linkcolor={red!50!black},
    citecolor={blue!100!black},
    urlcolor={black!80!black}
}

\newcommand{\defn}[1]{\emph{\textbf{{#1}}}}
\renewcommand{\paragraph}[1]{\vspace{.5 cm} \noindent \textbf{#1} }

\usepackage{bbm}

\newcommand{\mc}{\mathcal}

\newcommand{\nil}{\varnothing}

\newcommand{\interior}[1]{ {\kern0pt#1}^{\mathrm{o}} }

\newcommand{\eps}{\varepsilon}

\newcommand{\setof}[2]{\left\{ #1\; \colon \;#2 \right\}}
\newcommand{\set}[1]{\left\{ #1\right\}}

\newcommand{\C}{\mathbb{C}}

\newcommand{\F}{\mathbb{F}}
\newcommand{\K}{\mathbb{K}}
\newcommand{\N}{\mathbb{N}}

\newcommand{\NP}{\mathsf{NP}}

\usepackage{algorithm}
\usepackage[noend]{algpseudocode} 
\makeatletter
\newcounter{HALG@line}
\renewcommand{\theHALG@line}{\thealgorithm.\arabic{ALG@line}}
\makeatother

\usepackage[capitalise,nameinlink,noabbrev]{cleveref}
\crefname{equation}{}{} 
\crefname{enumi}{Step}{} 

\newtheorem{theorem}{Theorem}[section]
\newtheorem{lemma}[theorem]{Lemma}
\newtheorem{proposition}[theorem]{Proposition}

\newtheorem{claim}{Claim}[theorem]
\theoremstyle{definition}
\newtheorem{question}[theorem]{Question}
\newtheorem{remark}[theorem]{Remark}
\newtheorem{scheduler}[theorem]{Scheduler}

\usepackage{thmtools}
\usepackage{thm-restate}

\usepackage[
backend=biber,
style=alphabetic,
]{biblatex}
\setlist{nosep}
\setlist{nolistsep}
\setlist[enumerate]{leftmargin=*,parsep=0pt,label={\rm \alph*)}}
\setlist[itemize]{leftmargin=*,parsep=0pt}

\begin{document}

\newcommand{\B}{\mathsf{B}}
\newcommand{\alg}{\mathsf{alg}}
\newcommand{\opt}{\mathsf{opt}}
\newcommand{\tgr}{\mathsf{ins}}
\renewcommand{\K}{\mathsf{K}}
\renewcommand{\C}{\mathsf{C}}
\newcommand{\can}{\mathsf{nev}}
\newcommand{\eve}{\mathsf{eve}}
\newcommand{\tilC}{\widetilde{\C}}

\newcommand{\J}{{\mathcal{F}_\mathsf{big}}}
\newcommand{\A}{\mathcal{A}}
\renewcommand{\F}{\mathcal{F}}
\renewcommand{\S}{{\mathcal{F}_\mathsf{small}}}

\title{When to Give Up on a Parallel Implementation}
\author{Nathan S. Sheffield, Alek Westover}
\date{}
\maketitle
\abstract{
In the Serial Parallel Decision Problem (SPDP), introduced by Kuszmaul and Westover
[SPAA'24], an algorithm receives a series of tasks online, and must choose for each between a serial implementation and a parallelizable (but less efficient) implementation.
Kuszmaul and Westover describe three decision models: 
(1) \defn{Instantly-committing} schedulers must decide on arrival, irrevocably, which
implementation of the task to run.
(2) \defn{Eventually-committing} schedulers can delay their decision beyond a
task's arrival time, but cannot revoke their decision once made. 
(3) \defn{Never-committing} schedulers are always free to abandon their progress on the task and
start over using a different implementation. Kuszmaul and Westover gave a simple instantly-committing scheduler whose total completion time is $3$-competitive with the offline optimal schedule, and proved two lower bounds: no eventually-committing scheduler can have competitive ratio better than $\phi
\approx 1.618$ in general, and no instantly-committing scheduler can have
competitive ratio better than $2$ in general. They conjectured that the three decision models should admit different competitive ratios, but left upper bounds below $3$ in any model as an open problem.\\

In this paper, we show that the powers of instantly, eventually, and never
committing schedulers are distinct, at least in the ``massively parallel
regime''. The massively parallel regime of the SPDP is the special case
where the number of available processors is asymptotically larger than the number of tasks to process, meaning that the \emph{work} associated with running a task in serial is negligible compared to its \emph{runtime}.
In this regime, we show 
(1) The optimal competitive ratio for instantly-committing schedulers is $2$, 
(2) The optimal competitive ratio for eventually-committing schedulers lies in
$[1.618, 1.678]$, 
(3) The optimal competitive ratio for never-committing schedulers lies in $[1.366, 1.500]$. 
We additionally show that our instantly-committing scheduler is also $2$-competitive
outside of the massively parallel regime, giving proof-of-concept that results in the massively parallel regime can be translated to hold with fewer processors. 
}
\section{Introduction}
\subsection{Background}

Many computational tasks can be performed quickly in parallel over a large number of processors --- but such parallel implementations may be less work-efficient than a serial implementation on a single processor, requiring substantially more total computation time across all machines. When several different tasks must be completed in as little total time as possible, this trade-off between work and time can necessitate running different tasks in different modes: small tasks can be done in serial to save work, while large tasks must be parallelized to prevent their serial runtimes from dominating the overall computation.\\

To formalize this problem, Kuszmaul and Westover introduced the \defn{Serial Parallel Decision Problem (SPDP)}~\cite{ku24}. In their model, each task has exactly two possible implementations: a parallel implementation which can be worked on by multiple machines at once (where the rate of progress on the implementation is proportional to the number of processors assigned to it), and a serial implementation which can only be worked on by a single processor at a time.  If all tasks are available at time $0$, it is easy to efficiently determine the optimal strategy: all jobs with serial completion time smaller than some threshold can be run in serial, and the larger tasks must be run in parallel. The model becomes interesting when previously-unknown tasks are allowed to arrive at arbitrary times, and one wishes to minimize the competitive ratio between the total completion time of an \emph{online} algorithm compared to the offline optimal completion time.\\

Kuszmaul and Westover define three distinct versions of this model, parameterized by the degree to which the online scheduler is able to reverse its decisions. \\

\begin{enumerate}
    \item An \defn{instantly-committing} scheduler must choose an implementation for each task as soon as the task arrives, and is not allowed to revisit this choice.
    \item An \defn{eventually-committing} scheduler may delay choosing an implementation, but must choose one irrevocably before assigning its work to a processor.
    \item A \defn{never-committing} scheduler can, at any time, discard all as-yet completed work on an implementation and re-start the task with the other implementation.\\
\end{enumerate}

The distinction between the eventually- and never-committing models is motivated by potential practical concerns: if a task involves mutating an input in memory, it may not be feasible to cancel an implementation once it begins running. 
Westover and Kuszmaul present an instantly-committing scheduler achieving competitive ratio $3$, and show competitive ratio lower bounds of $2$ and $\phi \approx 1.618$ in the instantly-committing and eventually-committing models, respectively. They conjecture that the ability to delay or cancel choices should allow for more competitive online algorithms, but leave open the problem of finding better competitive ratio upper bounds than $3$.

\subsection{This Work}

In this work, we consider Kuszmaul and Westover's SPDP when the number of available processors is much larger than the number of tasks, noting that all of their upper and lower bounds hold in this parameter regime. This is a particularly simple setting, since the work associated to a serial implementation is now negligible compared to its completion time --- running a task in serial means accepting a lower bound on completion time, but requires essentially no work. We can think of this setting as an unrelated-machines scheduling problem with an unlimited number of identical ``slow'' machines, and a single unrelated ``fast'' machine, representing a massively parallel implementation of the task across many processors --- note that this could also describe scenarios with a literal fast machine, such as a single piece of accelerated hardware.\\

Our main results are tight bounds on the competitive ratio of
instantly-committing schedulers in this regime, and separations between the
strength of all 3 models. Our results are summarized in \cref{tab:result}.

\begin{table}[h!]
\centering
\begin{tabular}{|c|c|c|}
\hline
\textbf{Model}                   & \textbf{Lower Bound} & \textbf{Upper Bound} \\ \hline
Instantly-Committing Schedulers  & 2 \cite{ku24}     & \textbf{2} *  \\ \hline
Eventually-Committing Schedulers & 1.618 \cite{ku24}  & \textbf{1.678} *  \\ \hline
Never-Committing Schedulers      & \textbf{1.366} *      & \textbf{1.5} * \\ \hline
\end{tabular}
\caption{Main Results. * = this work.}
\label{tab:result}
\end{table}

In each case, the upper bound comes from a simple heuristic in which the
algorithm compares its projected completion time to its current estimate of the
optimal completion time. More precisely, at each time $t$, our scheduler
computes an offline optimal strategy ``$\opt(t)$'' on the truncation of the task
sequence to tasks that arrive before time $t$, and makes decisions based on the
completion time of $\opt(t)$. \\

Our main technical contribution is the analysis of the schedulers. Working with
$\opt(t)$ is challenging, because the schedules $\opt(t)$ and $\opt(t')$ can be
quite different for $t\neq t'$.
For instantly-committing schedulers we use an invariant-based approach to bound, at all times,
the work taken by our scheduler in terms of the minimum work and completion time among all schedulers.
For eventually- and never-committing schedulers this approach is no longer
feasible: there is no well defined notion of the ``work taken'' by our
scheduler, because the scheduler may not have committed to a decision yet.
Instead, these analyses rely on choosing a couple of critical times to observe
the state of our scheduler and $\opt(t)$, and then establishing a dichotomy: 
either (1) ``real fast tasks'' (tasks that $\opt$ runs on the fast machine)
arrive quickly, in which case our scheduler prioritizes real fast tasks on
the fast machine and will run most other tasks on slow machines, or (2) real
fast tasks arrive slowly, in which case our scheduler never falls too far behind
$\opt$, despite making suboptimal use of the fast machine.\\

In addition to these results, we show that with some effort our instantly-committing scheduler can be adapted to work for any number of processors, fully resolving the question of the optimal competitive ratio of instantly-committing schedulers in the general SPDP, and giving a proof-of-concept that results for a large number of processors can be adapted to hold when the work associated with serial tasks is also a concern.


\subsection{Related Work}
There is a long line of work studying the phenomenon of work-inefficient parallel implementations in multi-processor scheduling. Typically, the models of limited parallelism considered involve one of three types of jobs: \\

\begin{enumerate}
    \item \defn{Rigid} jobs, which come with a number $p_i$ specifying a fixed number of processors the job must be run on at each timestep of its execution.
    \item \defn{Moldable} jobs, where the scheduler may choose the (fixed) number of processors the job is run on, and the amount of work scales depending on this choice according to some \defn{speedup curve}.
    \item \defn{Malleable} jobs, which like moldable jobs have an associated speedup curve, but where the job may be assigned to different numbers of processors at different timesteps (as opposed to the scheduler choosing a fixed value at the start of the task's runtime).\\
\end{enumerate}

In each of these cases, there is interest in minimizing the total completion time (\defn{makespan}) in both the offline setting --- where problems tend to be $\NP$-hard, but may have approximation algorithms~\cite{Turek92, Mounie99, Tiwari94, Turek94_rigid, Turek94_moldable} --- and the online setting, where the goal is to minimize competitive ratio~\cite{Graham69, DuttonMao07, baker1983shelf, HurnikPaulus08, GuoKang10, Ye09,YeChen18}. Kuszmaul and Westover's Serial Parallel Decision Problem is related to this line of work, but doesn't quite fit into the usual framework --- in their model, instead of dealing with an arbitrary speedup curve, there is a single binary decision between a completely serial and perfectly parallelizable implementation.\\

As noted, the massively-parallel regime of the SPDP considered in this paper can be naturally viewed as a scheduling problem with an unlimited number of identical ``slow'' machines, and a single unrelated ``fast'' machine. Standard scheduling problems in the unrelated machines model have also been well-studied, in terms of both offline approximation algorithms and hardness results~\cite{horowitz76exact, lenstra1990approximation, ss-randomized-rounding, two-types-unrelated, molinaro19stochastic, page20intersection, gupta2021stochastic, im2023improved, deng2023generalized}, and online algorithms~\cite{awerbuch1995load, aspnes1997line, ss-randomized-rounding, caragiannis2008better, anand2011meeting, gupta2017stochastic, gupta2020greed, zhang2022randomized}. We note, however, that since we treat ``slow'' machines as an unbounded resource, and there is only a single fast machine available, most of the typical difficulties of multi-processor scheduling problems do not arise. In particular, unlike a typical load-balancing problem where $\NP$-hardness follows from a standard set-partition reduction, the One-Fast-Many-Slow Decision Problem (without dependencies) is easily solvable offline, simply by putting all tasks which finish below a certain threshold on distinct slow machines.

\subsection{Open Questions}
We leave three main open questions as directions for future work.
\begin{question}
What are the optimal competitive ratios for eventually/never-committing schedulers?
\end{question}

In \cref{appendix} we identify barriers, showing that improving on our
eventually/never-committing schedulers will require substantially different
algorithms --- but we suspect that such improvements may be possible.\\

\begin{question}
Are randomized schedulers more powerful than deterministic schedulers?
\end{question}
In the main body of the paper we consider only deterministic schedulers;
however, for many online problems randomized algorithms can do substantially
better than deterministic ones. In \cref{sec:lower} we give some lower bounds
against randomized schedulers, but these bounds are weaker than those known for deterministic schedulers.\\

\begin{question}
Is there a general transformation between schedulers for the massively parallel
regime (i.e. the One-Fast-Many-Slow Decision Problem) and the general SPDP?
\end{question}

The fundamental difficulty of the SPDP is deciding between implementations which take a lot of work, and implementations which take a lot of time. This tradeoff is absolute in the massively parallel regime, since the large number of processors means the amount of work associated with a serial implementation is negligible, whereas in the general SPDP it is possible for all processors to be saturated with serial implementations to run. Intuitively, one might expect that having work associated to the serial implementations only makes the problem \emph{easier}, since it makes the tradeoff less dramatic --- indeed, Kuszmaul and Westover's competitive ratio lower bounds become weaker when the number of processors is small. So, one might hope that algorithms in the massively parallel regime can be generically translated to limited-processor settings. Formalizing this connection is an interesting direction for future research.

\section{Preliminaries}
\subsection{The One-Fast-Many-Slow Decision Problem}

In this section we formally define the \defn{One-Fast-Many-Slow Decision Problem}, where the goal is to distribute work between a single \defn{fast machine} and an unlimited number of \defn{slow machines}. An instance of the problem is a \defn{Task Arrival Process (TAP)} $\mc{T} = (\tau_1, \dots, \tau_n)$, where each task $\tau_i$ consists of a tuple $(\sigma_i, \pi_i, t_i)$ indicating runtime on a slow machine, runtime on the fast machine, and arrival time, respectively, such that $t_1 \leq \dots \leq t_n$. A valid schedule associates at most one task to each machine at each point in time\footnote{In order for the notions like ``amount of work performed on $\tau_i$'' to be well-defined, we must additionally mandate that a schedule be measurable. Alternatively, one can assume that time is discretized into appropriately fine timesteps.} such that no work is done on any task before its arrival time, each task runs on at most 1 machine, and each task $\tau_i$ is either run for a total of $\sigma_i$ time on some slow machine, or a total of $\pi_i$ time on the fast machine. The \defn{completion time} (also known as \defn{makespan}) of the schedule is the time when the last task is finished.\\

We will be interested in online algorithms for this problem. An online scheduler learns about each task only at its arrival time, and at each time $t$ must already have fixed the prefix of the schedule on times less than $t$. We define three distinct models for how these online decisions are made: \\
\begin{enumerate}
    \item For each task $\tau$, an \defn{instantly-committing} scheduler must fix at $\tau$'s arrival time the machine that $\tau$ will run on.
    \item An \defn{eventually-committing} scheduler need not fix a machine for any task until that task begins running.
    \item A \defn{never-committing} scheduler is an eventually-committing scheduler with the additional power to, at any time, ``cancel'' a task from the schedule, erasing all work previously done on the task and allowing it to be re-assigned to a new machine.\\
\end{enumerate}

In each case, we are interested in minimizing the \defn{competitive ratio} of an online scheduler, which is the supremum over all TAPs of the ratio of the online scheduler's completion time to the completion time of an optimal scheduler on that TAP.


\subsection{Connection to the SPDP}

In Kuszmaul and Westover's Serial Parallel Decision Problem, a scheduler must allocate work to $p$ equally-powerful processors, where each task is specified by the work of the serial implementation, the work of the parallel implementation, and the runtime. The scheduler must choose whether to run each task in serial or parallel, and then must assign the resulting work to the $p$ processors, where parallel work can run on multiple processors at once but serial work cannot. \\

We can define the \defn{massively parallel} regime of this problem to be the limit as the number of processors becomes large compared to the number of tasks. Letting $n$ be the number of tasks, if $n \leq \varepsilon p$ then restricting the serial implementations to run on only the first $n$ many processors, and the parallel implementations to run on only the last $p-n$ many processors, the completion time can increase by at most a $\frac{1}{1-\varepsilon}$ factor. This corresponds directly to the One-Fast-Many-Slow Decision Problem: we think of each of these serial processors as a ``slow machine'', noting that since we have as many as we have tasks there are effectively an unlimited number. We think of the parallel processors collectively as a ``fast machine'', noting that we can assume without loss of generality that, at any point in time, all parallel processors are running the same parallel implementation.

\subsection{Notation}

We now introduce our notation for describing and analyzing schedulers. For algorithm $\alg$ and TAP $\mc{T}$, we let $\C_\alg^{\mc{T}}$ be the
completion time of $\alg$ on $\mc{T}$. Let $\mc{T}^{t}$ be the truncation of TAP $\mc{T}$ consisting of the tasks $\tau_i$ with $t_i\le t$. When $\mc{T}$ is clear from context we will write $\C_\alg^{t}$ to denote $\C_\alg^{\mc{T}^{t}}$, and we will write $\C_\alg$ to denote $\C_{\alg}^\infty = \C_\alg^{\mc{T}}$. We will also use $\widetilde{\C}_\alg$ to denote the completion time of the fast machine --- that is, the final time when the fast machine has work.\\

It will be useful to be able to talk about the optimal completion time of a prefix of the TAP. Define the schedule $\opt(t)$ to be a schedule for $\mc{T}^{t}$ with minimal completion time. Note that $\opt(t)$ is only defined as an offline strategy, but that an online algorithm can compute it at time $t$, thus obtaining a lower bound on $\C_\opt$, which will be useful to inform the algorithm's future decisions. For ease of notation, we'll often abbreviate $\C_{\opt(t)}^{t}$ as $\C^t$.\\

There may be many sets of decisions which result in the optimal completion time; as opposed to letting $\opt(t)$ be an arbitrary such scheduler, it will be useful to fix a canonical one, which we will do by letting $\opt(t)$ run as many tasks in serial as possible.

\begin{scheduler}\label{alg:opt} The scheduler $\opt(t)$, defined on
$\mc{T}^{t}$, makes decisions as
follows:
\begin{itemize}
\item If $\tau_i$ has $\sigma_i+t_i \le \C^{t}$, run $\tau_i$ on a
slow machine when it arrives.
\item Otherwise, run $\tau_i$ on the fast machine. Prioritize tasks with larger $\sigma_i+t_i$, and break ties by taking tasks with smaller $i$.\\
\end{itemize}
\end{scheduler}

Finally, we let $[n]=\{1, \dots, n\}$, and for a set $J$ of tasks we will write $\pi_J$ to denote $\sum_{j \in J} \pi_j$. 

\section{A $2$-Competitive Instantly-Committing Scheduler}
In this section we present and analyze a $2$-competitive instantly-committing scheduler. Kuszmaul and Westover showed that a competitive ratio of $(2-\eps)$ is impossible for instantly-committing schedulers, so our scheduler is optimal.
The scheduler, which we call $\tgr$ is defined in \cref{alg:2gr}.

\begin{scheduler}\label{alg:2gr} When task $\tau_i$ arrives:
\begin{itemize}
\item If $\sigma_i+t_i > 2\C^{t_i}$ run $\tau_i$ on the fast machine, with the fast machine processing tasks in order of arrival.
\item Otherwise run $\tau_i$ on a slow machine.\\
\end{itemize}
\end{scheduler}

We analyze $\tgr$ by showing inductively that $\widetilde{\C}_\tgr$ (the completion time of the fast machine) is small compared to the work and completion time of \emph{any} other schedule. For length $n$ TAP $\mc{T}$, scheduler $\alg$, and $i\in [n]$, we define the quantity $\K_\alg^{t}$ to be the sum of $\pi_j$ for all tasks $\tau_j\in \mc{T}^{t}$ that $\alg$ runs on the fast machine. 
The key to analyzing $\tgr$ is the following lemma.

\begin{lemma}\label{lem:tgrinvar}
Fix a length $n$ TAP. For all $i\in [n]$, and for all instantly-committing schedulers $\alg$,  
\begin{equation}\label{eq:tgrinvar}
\widetilde{\C}_\tgr^{t_i} \le \C_\alg^{t_i} + \K_\alg^{t_i}.
\end{equation}
\end{lemma}
\begin{proof}
We prove the lemma by induction on $i$.
For $i=1$ the claim is trivial.
Now, fix $i \in [n-1], \alg$ and assume the lemma for $i$ and for all $\alg'$;
we will prove the lemma for $i+1, \alg$.\\

If $\tgr$ runs $\tau_{i+1}$ on a slow machine then
$\widetilde{\C}_\tgr^{t_{i+1}}=\widetilde{\C}_{\tgr}^{t_i}$, 
and $\C_\alg^{t_i}+\K_\alg^{i} \le \C_\alg^{t_{i+1}}+\K_\alg^{i+1}$. 
Thus, the invariant \cref{eq:tgrinvar} is maintained. 
We always have $\widetilde{\C}^{t_{i+1}}_\tgr \le \widetilde{\C}^{t_{i}}_\tgr + \pi_{i+1}$, so if $\alg$ runs $\tau_{i+1}$ on the fast machine then the invariant \cref{eq:tgrinvar} is also maintained, since then $\C_\alg^{t_{i+1}} + \K_\alg^{t_{i+1}} \geq \C_\alg^{t_{i}} + \K_\alg^{t_{i}} + \pi_{i+1} \geq \widetilde{\C}^{t_{i}}_\tgr + \pi_{i+1}$ by the inductive hypothesis.\\

The final case to consider is when $\alg$ runs $\tau_{i+1}$ on a slow machine, while
$\tgr$ runs $\tau_{i+1}$ on the fast machine. 
From the definition \cref{alg:2gr} of $\tgr$, the fact that $\tgr$ ran ran
$\tau_{i+1}$ on a slow machine implies 
\begin{equation}
  \label{eq:thistaskisbig}
2\C^{t_{i+1}} < \sigma_{i+1}+t_{i+1} .
\end{equation}
On the other hand, $\alg$ ran $\tau_{i+1}$ on the fast machine. 
Thus, 
\begin{equation}
  \label{eq:algissuperbig}
\sigma_{i+1}+t_{i+1}\le \C_\alg^{t_{i+1}}.
\end{equation}
Now, we use the invariant for $(i,\opt_{t_{i+1}})$ to bound
$\widetilde{\C}_{\tgr}^{t_{i+1}}$
. We have:
\begin{equation}
  \label{eq:tgrtiplusone}
\widetilde{\C}_\tgr^{t_{i+1}} \le \widetilde{\C}_\tgr^{t_i}  + \pi_{i+1} \le \K_{\opt(t_{i+1})}^{t_i}+\C_{\opt(t_{i+1})}^{t_i} + \pi_{i+1}.
\end{equation}
Because of \cref{eq:thistaskisbig} we know that $\opt(t_{i+1})$ must
run $\tau_{i+1}$ on the fast machine. So, we have
\begin{equation}
  \label{eq:lasttgreqn}
\K_{\opt(t_{i+1})}^{t_i}+\C_{\opt(t_{i+1})}^{t_i} + \pi_{i+1} = \K_{\opt(t_{i+1})}^{t_{i+1}}+\C_{\opt(t_{i+1})}^{t_i} \leq 2 \C^{t_{i+1}}.
\end{equation}
Stringing together the above inequalities \cref{eq:tgrtiplusone}, \cref{eq:lasttgreqn}, \cref{eq:thistaskisbig}, and \cref{eq:algissuperbig}, we get
\[  \widetilde{\C}_\tgr^{i+1} < \C_\alg^{t_{i+1}}. \]
Thus, the invariant \cref{eq:tgrinvar} holds.
\end{proof}

Using \cref{lem:tgrinvar} it is easy to show that $\tgr$ is $2$-competitive.
\begin{theorem}\label{thm:tgr}
$\tgr$ is a $2$-competitive instantly-committing scheduler.
\end{theorem}
\begin{proof}
By \cref{lem:tgrinvar} we have $\widetilde{\C}_{\tgr}^{t_n} \le 2\C_\opt$.
Thus, $\tgr$ finishes using the fast machine before time $2\C_\opt$.
Any task that $\tgr$ runs on a slow machine must have
$ \sigma_i+t_i \le 2\C_\opt, $
so these tasks finish before $2\C_\opt$ as well.
\end{proof}

\section{A $1.678$-Competitive Eventually-Committing Scheduler}\label{sec:eventual}
In this section we present and analyze a $\xi$-competitive eventually-committing scheduler, where $\xi\approx 1.678$ is the real root of the polynomial $2x^3-3x^2-1$. Kuszmaul and Westover gave a lower bound of $\phi\approx 1.618$ on the competitive ratio of any eventually-committing scheduler and conjectured that this lower bound is tight. Our scheduler represents substantial progress towards resolving Kuszmaul and Westover's conjecture, improving on their previous best algorithm which had a competitive ratio of $3$.
Our scheduler, which we call $\eve$, is defined in \cref{alg:eve}.

\begin{scheduler}\label{alg:eve} At time $t$:
\begin{itemize}
\item If task $\tau_i$, which has arrived but not yet been started, has
$\sigma_i+t \le \xi \C^t$, then start $\tau_i$ on a slow machine.
\item Maintain up to one \defn{active} task at a
time. The fast machine is always allocated to the active task.
\item When there is no active task, but there are unstarted tasks present, choose as the new active task the unstarted task with the largest $\sigma_i+t_i$ value (breaking ties arbitrarily).\\
\end{itemize}
\end{scheduler}

\begin{theorem}
$\eve$ is a $\xi$-competitive eventually-committing scheduler.
\end{theorem}
\begin{proof}
Fix TAP $\mc{T}$.
Let $\tilC_\eve$ denote the time when $\eve$ completes the last task run on the
fast machine.
If $\tau_i$ is run on a slow machine at any time $t$, then $\tau_i$ finishes before
$\xi \C^t \le \xi\C_\opt$. Thus, it suffices to show that $\tilC_\eve\le \xi\C_\opt.$\\

For any $x\in [0,\C_\opt]$, let $R(x)$ be the first time that an online algorithm becomes aware that the optimal schedule requires at least $x$ completion time --- that is, $R(x) = \inf\setof{t}{\C^{t}\ge x}$.
Let $\A$ (``actual'') be the set of tasks that $\opt$ runs on the fast machine, and $\F$ (``fake'') be the set of tasks that $\opt$ runs on a slow machine but $\eve$ runs on the fast machine. We can bound the sizes and arrival times of tasks in $\F$ as follows.

\begin{claim}\label{clm:junk-comes-early}
All tasks $\tau_i\in \F$ arrive before time $R(\C_\opt/\xi)$, and have
$\pi_i < \sigma_i/\xi$.
\end{claim}
\begin{proof}
All tasks $\tau_i \in \F$ are run on the fast machine by $\eve$, and on slow
machines by $\opt$. 
In particular this means
\[ \xi\C^{t_i}< \sigma_i+t_i \le \C_\opt \le \xi \C^{R(\C_\opt/\xi)}. \] 
Thus, $t_i< R(\C_\opt/\xi)$. To show $\pi_i < \sigma_i/\xi$, note that $ \pi_i+t_i \le \C^{t_i} < \frac{\sigma_i + t_i}{\xi}$. \\
\end{proof}

To analyze when tasks in $\F$ get run it will be useful to partition $\F$ into $\J = \{\tau_i \in \F \colon \sigma_i + t_i > \C_\opt / \xi\}$ and $\S = \{\tau_i \in \F \colon \sigma_i + t_i \leq \C_\opt / \xi\}$. 
Now we show that, without loss of generality, $\eve$ does not start any tasks in $\S$ too late.

\begin{claim}\label{clm:junk-free-zone}
If $\eve$ starts a task $\tau\in \S$ at any time $t\ge R(\C_\opt/\xi)$, then $\tilC_\eve\le \xi \C_\opt$.
\end{claim}
\begin{proof}
Note that no task $\tau_i \in \S$ can be started after time $\C_{\opt}$: since
$\C_{\opt} + \sigma_i \le \C_\opt + \C_\opt / \xi < \xi \C_\opt$, any task
$\tau_i \in \S$ present but not already running at time $\C_\opt$ would be run
on a slow machine. Let $t_*$ be the last time after $R(\C_\opt/\xi)$ when $\eve$
starts a task $\tau\in\F$. If $t_*$ does not exist the claim is vacuously true.
In light of our previous observation, $t_*< \C_\opt$.
Let $\tau_i$ be the task that $\eve$ starts at time $t_*$.
Because $\eve$ prioritizes making tasks with larger $\sigma_j+t_j$ values
active, at time $t_*$ there are no tasks $\tau\in \J\cup \A$ present.
After time $t_*$, no more tasks from $\F$ can arrive by \cref{clm:junk-comes-early}, and at most $\C_\opt - t_*$ work in $\A$ can arrive because $\opt$ must be able to complete this work.
Thus, 
\begin{equation}\label{eveeveveveeveve}
\tilC_\eve \le t_* + (\C_\opt - t_*) + \pi_i = \C_\opt + \pi_i. 
\end{equation}
Now, because $\tau_i\in \S$ we have $\pi_i\le \C_\opt/\xi^2$; using this in \cref{eveeveveveeveve} we find $\tilC_\eve\le (1+1/\xi^2)\C_\opt\le \xi\C_\opt$.
\end{proof}
This means that we can assume that, after time $R(\C_\opt/\xi)$, the only tasks that $\eve$ runs on the fast machine are $\A$, $\J$, and whatever the active task was at time $R(\C_\opt/\xi)$. We call the active task at time $R(\C_\opt/\xi)$, if one exists, the \defn{stuck} task, denoted $\tau_s$. 
We split into cases depending on how large this stuck task is.

\paragraph{Case 1:} There is no stuck task.\\

In this case, we in fact have $\widetilde{C}_\eve \le \C_\opt$. Since there is no active task at time $R(\C_\opt/\xi)$, there are no tasks present but not started on slow machines. By \cref{clm:junk-comes-early}, $\eve$ will run all tasks $\tau \not\in \A$ arriving after time $R(\C_\opt/\xi)$ on slow machines. Thus, at all time steps $t\in [R(\C_\opt/\xi),\C_\opt]$, $\eve$ either has no active task on the fast machine, or has some $\tau\in \A$ as the active task on the fast machine, so $\eve$ completes $\A$ at least as quickly as $\opt$.

\paragraph{Case 2:} There is a stuck task, with $\sigma_s + t_s > \C_\opt/\xi$.\\

Define $\A_{\text{early}} = \setof{\tau_i\in \A}{t_i < R(\C_\opt/\xi)}$ and $\A_{\text{late}} = \A\setminus \A_{\text{early}}$.
Let $t< R(\C_\opt/\xi)$ be a time when all tasks in $\set{\tau_s}\cup \J \cup \A_{\text{early}}$ have already arrived; such a time must exist by \cref{clm:junk-comes-early}.
Observe that $\opt(t)$ runs all tasks in $\set{\tau_s}\cup \J\cup\A_{\text{early}}$ on the fast machine due to $\C^t < \C_\opt/\xi$. This further implies that $\pi_{\J\cup \set{\tau_s}\cup\A_{\text{early}}}\le \C_\opt/\xi$.
Also $\pi_{\A_{\text{late}}}\le \C_\opt - R(\C_\opt/\xi)$, simply because $\opt$ must complete the work on tasks $\A_{\text{late}}$ after these tasks arrive. By \cref{clm:junk-free-zone} we may assume without loss of generality that,
after time $R(\C_\opt/\xi)$, $\eve$ is always running a task from $\set{\tau_s}\cup \A\cup \J$. Thus,
\begin{align*}
    \widetilde{\C}_\eve &\le R(\C_\opt/\xi) + \pi_{\A \cup \J \cup \{\tau_s\}}\\
    & = R(\C_\opt/\xi) + \pi_{\A_{\text{early}} \cup \J \cup \{\tau_s\}} + \pi_{\A_{\text{late}}}\\
    & \le R(\C_\opt/\xi) + \C_\opt/\xi + (\C_\opt - R(\C_\opt/\xi))\\
    & \le \xi \C_\opt
\end{align*}

\paragraph{Case 3:} There is a stuck task, with $\sigma_s + t_s \le C_\opt / \xi$.\\

This case will be the most difficult to handle of the three.
It will be useful to focus now on the tasks of $\J$ that arrive after the stuck task is started. Let $t_*$ be the time when $\eve$ starts running $\tau_s$, and let $\J' = \setof{\tau_i \in \J}{t_i \ge t_*}$ be the fake tasks arriving after $t_*$. We first observe that, if no such tasks arrive, $\eve$ performs very well.

\begin{claim}\label{clm:noJ}
If $\J'=\nil$ then $\tilC_\eve \le \xi\C_\opt$.
\end{claim}
\begin{proof}
At time $t_\star$ no tasks $\tau_i\in \mc{A}\cup \J$ can be present, since all such tasks have $\sigma_i + t_i > \C_\opt/\xi$, so $\eve$ would prioritize running them on the fast machine instead of the stuck task. By \cref{clm:junk-free-zone}, we know that after time $t_*$ $\eve$ will always be running tasks from $\set{\tau_s}\cup \A \cup \J'$. The total work on tasks from $\mc{A}$ that arrives after time  $t_\star$ is at most $\C_\opt-t_\star$, so if $\J' = \nil$ we have
\[\tilC_\eve \le t_\star+\pi_s+\C_\opt-t_\star \le \sigma_s/\xi + \C_\opt \le \xi\C_\opt.\]
\end{proof}
By \cref{clm:noJ} we may assume $\J'\neq\nil$. So, let $\sigma_{\min}=\min_{\tau_i\in \J'}(\sigma_i)$; we will be able to control how much work arrives in the TAP by the fact that $\eve$ never decides to run the task $\tau_i\in \J'$ with $\sigma_i =\sigma_{\min}$ on a slow machine. Split $\A$ into $\A_{\text{early}} = \setof{\tau\in\A}{t_\star \le t_i < R(\sigma_{\min})}$ and $\A_{\text{late}} = \setof{\tau\in\A}{t_i \ge \max(t_\star, R(\sigma_{\min}))}$ (note that we use a different threshold to define earliness here than we did in  case 2). First we need the following analogue of \cref{clm:junk-free-zone}.

\begin{claim}\label{clm:big-junk-free-zone}
If $\eve$ starts a task $\tau\in \J'$ at any time $t\in [R(\C_\opt/\xi), \C_\opt]$, then $\tilC\le \xi\C_\opt$.
\end{claim}
\begin{proof}
Let $t_*$ be the last time in $[R(\C_\opt/\xi),\C_\opt]$ when $\eve$
starts a task $\tau\in\F'$. If $t_*$ does not exist the claim is vacuously true.
Let $\tau_i$ be the task that $\eve$ starts at time $t_*$.
Because $\eve$ prioritizes making tasks with larger $\sigma_j+t_j$ values
active, at time $t_*$ there are no tasks $\tau\in \A$ present.
After time $t_*$ at most $\C_\opt - t_*$ work in $\A$ can arrive because $\opt$ must be able to complete this work.
Recalling that $\pi_{\J'} \le \C_\opt/\xi$, we have
\[ \tilC_\eve \le t_* + \pi_{\J'} + (\C_\opt - t_*) \le \xi \C_\opt. \]
\end{proof}

\begin{claim}\label{clm:boundsomething}
$t_\star+\pi_{\J'} + \pi_{\A_{\mathsf{early}}} < (\sigma_{\min}+R(\sigma_{\min}))/\xi$.
\end{claim}
\begin{proof}
Fix a time $t<R(\sigma_{\min})$ after all tasks in $\J'\cup\A_{\mathsf{early}}$ have arrived, and fix a task $\tau_i \in \J'$ with $\sigma_i = \sigma_{\min}$. First, note that by \cref{clm:big-junk-free-zone} we can assume that $\eve$ has not started $\tau_i$ by time $t$. Thus, $\eve$ is free to start $\tau_i$ on a slow machine, but chooses not to. This implies
\begin{equation}
  \label{eq:sigmatxiC}
    \xi\C^{t} < \sigma_i + t < \sigma_{\min} + R(\sigma_{\min}).
\end{equation}

We also observe that $\opt(t)$ must run $\J'\cup \A_{\mathsf{early}}$ on the fast machine, since running any of them on the slow machine would finish after time $\sigma_{\min}$.
Thus, 
\begin{equation}
  \label{eq:tpiJpAsm}
 t_\star + \pi_{\J'} + \pi_{\A_{\mathsf{early}}} \le \C^t.
\end{equation}

Combining \cref{eq:tpiJpAsm} and \cref{eq:sigmatxiC} gives the desired statement.
\end{proof}

The other observation we make is that $R(\sigma_{\min})$ cannot happen too early.

\begin{claim}\label{clm:Rsigmamin-late}
$R(\sigma_{\min}) > (\xi - 1)\sigma_{\min}$.
\end{claim}
\begin{proof}
Let $\tau_i \in \J'$ be a task with $\sigma_i = \sigma_{\min}$. 
Note that $t_i \le R(\sigma_{\min})$ by \cref{clm:junk-comes-early}. 
Then, by \cref{clm:big-junk-free-zone} we may assume without loss of generality that at time $R(\sigma_{\min})$ $\eve$ is not running $\tau_i$, and does not choose to start $\tau_i$ on a slow machine. So, 
\[R(\sigma_{\min}) + \sigma_{\min} > \xi \C^{R(\sigma_{\min})} \ge \xi \sigma_{\min}.\]
\end{proof}

Now we conclude the theorem.
\begin{claim}
$\tilC_\eve\le \xi\C_\opt$.
\end{claim}
\begin{proof}
Note that $\pi_{\A_{\text{late}}}\le \C_\opt-R(\sigma_{\min})$.
Also, note that since $\tau_s$ was not put on a slow machine immediately upon arrival, we must have $\pi_s \le \sigma_s/\xi \le \C_{\opt}/\xi^2$.
Then, applying
\cref{clm:boundsomething} and \cref{clm:Rsigmamin-late} we have
\begin{align*}
\tilC_\eve&\le \pi_s + t_\star + \pi_{\J'}+ \pi_{\A_{\mathsf{early}}} + \pi_{\A_{\text{late}}} \\
&\le \C_\opt/\xi^2 + (\sigma_{\min}+R(\sigma_{\min}))/\xi + \C_\opt-R(\sigma_{\min})\\
&\le (1/\xi^2 + (1+\xi-1)/\xi + 1-(\xi-1))\C_\opt\\
&= (3 + 1/\xi^2 - \xi)\C_\opt\\
&= \xi\C_\opt.
\end{align*}
\end{proof}
\end{proof}

\begin{remark}
The simple nature of the lower bound in \cref{lb:phi}, along with the fact that
$\eve$ gets a competitive ratio quite close to $\phi$ might leave the impression
that $\phi$ is clearly the correct competitive ratio, and a slightly better
analysis of (a natural variant of) $\eve$ would be $\phi$-competitive. However,
this is not the case: in \cref{appendix}, we show that no \defn{non-procrastinating} eventually-committing scheduler can achieve competitive ratio better than $\xi$, where a scheduler is called non-procrastinating if, whenever tasks are present, it always runs at least one task. Thus, if the competitive ratio of $\eve$ can be improved upon, doing so will require a substantially different scheduler: one which occasionally decides to do nothing at all despite work being present.
\end{remark}

\section{A $1.5$-Competitive Never-Committing Scheduler}\label{sec:never}
\newcommand{\JJ}{\mathcal{J}}
In this section we analyze never-committing schedulers.
First we give a simple lower bound.
\begin{proposition}\label{prop:canclb}
Fix $\eps>0$. There is no deterministic $((1+\sqrt{3})/2-\eps)$-competitive
never-committing scheduler.
\end{proposition}
\begin{proof}
We may assume $\eps<.001$.
Let $\psi = (\sqrt{3}-1)/2$.
Let $\mc{T} = \Big((1, 2\psi, 0), (\infty, 1-\psi, \psi) \Big)$; that is, $\tau_1$ has $\sigma_1=1,\pi_1=2\psi,t_1=0$, and $\tau_2$ has $\sigma_2=\infty, \pi_2=1-\psi, t_2=\psi$.
Let $\mc{T}' = \Big((1, 2\psi, 0)\Big)$ be the same TAP without the second task. We have $\C_{\opt}^{\mc{T}'} = 2\psi$, since $\opt$ just runs the single task on the fast machine, and $\C_\opt^\mc{T} = 1$, since $\opt$ can run $\tau_1$ on a slow machine and $\tau_2$ on the fast machine. \\ 

Suppose that $\alg$ is a $(\psi+1-\eps)$-competitive scheduler. On TAP $\mc{T}'$, at time $\psi - \eps/2$, we claim $\alg$ must be running $\tau_1$ on the fast machine. If not, then $\alg$'s completion time must be at least $\min(\sigma_1, \psi - \eps/2 + \pi_1) = 1$, with the branch of the min depending on whether $\tau_1$ is ever moved to the fast machine --- but this gives competitive ratio $1/(2\psi) = 1+\psi$. \\

Before time $\psi$, it is impossible to distinguish between
$\mc{T}$ and $\mc{T}'$. Thus, $\alg$ must be running $\tau_1$ on the fast machine at time $\psi - \eps/2$ on TAP $\mc{T}$. Now, we have $\C_{\alg}^{\mc{T}}\ge \min(\sigma_1 + \psi - \eps/2, \pi_1 + \pi_2) = 1 + \psi - \eps/2$, with the branch of the min depending on whether $\tau_1$ is ever moved to a slow machine --- but this gives competitive ratio $(1+\psi - \eps/2)/1$. Thus, $\alg$ is not actually $(\psi+1-\eps)$-competitive.\\
\end{proof}

Now we give a $1.5$-competitive never-cancelling scheduler, which we call $\can$. Note that this competitive ratio is smaller than the lower bound of $\phi \approx 1.618$ known for eventually committing schedulers, so this demonstrates a separation between the strengths of schedulers in the two models.

\begin{scheduler}\label{alg:cancel} At time $t$:
\begin{itemize}
\item If task $\tau_i$ has $\sigma_i + t\le 1.5 \C^{t}$ but is not currently running on a slow machine, start $\tau_i$ on a slow machine, cancelling its fast machine implementation if necessary.
\item Let $\mc{P}$ be the set of $\tau_i$ that have arrived and are not running on a
slow machine. Choose $\tau_i\in \mc{P}$ maximizing $\sigma_i + t_i$, breaking ties
by choosing the task with the smaller $i$. Run $\tau_i$ on the fast machine during this time step.
\end{itemize}
\end{scheduler}



\begin{theorem}
$\can$ is a $1.5$-competitive never-committing scheduler. 
\end{theorem}
\begin{proof}
Fix TAP $\mc{T}$. 
Let $\widetilde{C}_\can$ denote the final time when $\can$ has work on the fast machine.
Observe that if $\can$ ever runs $\tau_i$ on a slow machine, then $\can$
finishes $\tau_i$ before time $1.5\C_\opt$. Thus, to show that $\can$ is
$1.5$-competitive it suffices to show $\widetilde{C}_\can\le 1.5\C_\opt$. 
\\

Let $\A = \{\tau_i\colon \sigma_i + t_i > \C_\opt\}$ be the set of tasks that $\opt$ \emph{actually} runs on the fast machine. 


\begin{claim}\label{clm:wedoreal}
$\can$ never runs a task $\tau\in \A$ on the fast machine after time $\C_\opt$.
\end{claim}
\begin{proof}
$\can$ always allocates the fast machine to the present task with the largest
value of $\sigma_i + t_i$ among tasks that aren't running on slow machines.
Thus, whenever there are tasks from $\A$ that aren't running on slow machines, $\can$ will run one such task on the fast machine. $\opt$ is able to complete all tasks in $\A$ on the fast machine by time $\C_\opt$. Thus, $\can$ completes or starts on slow machines all tasks $\tau\in \A$ before time $\C_\opt$.
\end{proof}

This means that the only way to have $\tilC_\eve > \C_\opt$ is if there are tasks with $\sigma_i + t_i \le \C_{\opt}$ that have yet to be completed at time $\C_{\opt}$; we assume that this is the case for the remainder of the proof. 
Let $\Pi(x)$ be the total amount of work $\can$ performs on the fast machine after time $\C_\opt$ across all tasks with $\sigma_i + t_i \in [x, 1.5x]$. 
For any $x \le \C_{\opt}$, let $R(x) = \inf \{t \ | \ \C^t \ge x\}$ be the first time an online algorithm becomes aware that the optimal schedule requires $x$ completion time; the following key claim allows us to bound this left-over work $\Pi(x)$ in terms of $R(x)$.

\begin{claim}\label{clm:leftover-work-small}
For all $x$, we have $\Pi(x) \le x - R(x)$.
\end{claim}
\begin{proof}
Let $\mc{J}_x$ denote the set of tasks with $\sigma_i+t_i\le 1.5x$ that $\can$
runs on the fast machine at some time after $\C_\opt$. First, note that all
$\tau_i\in \mc{J}_x$ must have $t_i < R(x)$, or else $\tau_i$ would be placed on a
slow machine upon arrival.
Choose $\varepsilon$ sufficiently small, such that no tasks arrive between times
$R(x) - \varepsilon$ and $R(x)$. Since $R(x) - \varepsilon < R(x)$, we know
$\C^{R(x) - \varepsilon} < x$, and so $\opt(R(x) - \varepsilon)$ must run all
tasks with $\sigma_i + t_i \ge x$ on the fast machine. In order for $\opt(R(x) -
\varepsilon)$ to finish these tasks before time $\C^{R(x) - \varepsilon} < x$,
$\opt(R(x) - \varepsilon)$ must have at most $x - R(x) + \varepsilon$ fast work
remaining across all such tasks. \\

Now, by the same argument as in \cref{clm:wedoreal}, because $\can$ prioritizes
tasks with $\sigma_i+t_i\ge x$ over tasks with $\sigma_i+t_i< x$ on the fast
machine whenever they are present (and not yet started on slow machines), $\can$
has at most $x - R(x) + \varepsilon$ work remaining on tasks in $\mc{J}_x$ at
time $R(x) - \varepsilon$. Because no more tasks from $\mc{J}_x$ arrive after
this time, we have $\Pi(x) \le x-R(x)+\eps$ as well. The claim held for all
$\eps>0$, and so taking $\eps\to 0$ we have $\Pi(x)\le x-R(x)$.
\end{proof}

We now give an observation to control $R(x)$. 
Let $\tau_{i_\star}$ be the task, among all tasks that $\can$ runs on the fast machine after time $\C_\opt$, with the smallest value of $\sigma_i + t_i$.
Let $\lambda = \sigma_{i_\star}+t_{i_\star}$.

\begin{claim}\label{clm:learn-late}
For all $x \ge \lambda$, we have $R(x) > 1.5 x - \lambda$.
\end{claim}
\begin{proof}
First, note that $t_{i_\star}\le R(\lambda)\le \lambda$ or else $\can$ would start
$\tau_{i_\star}$ on a slow machine upon arrival.
Now, because $\can$ doesn't start $\tau_{i_\star}$ on a slow machine
at time $R(x)> t_{i_\star}$, we have $R(x) + \sigma_{i_\star} > 1.5 x$.
\end{proof}

To prove the theorem, it will now suffice to branch into two cases, based on how large $\lambda$ is.

\paragraph{Case 1:} $\lambda \ge (2/3) \C_{\opt}$.\\

In this case, since $1.5 \lambda \ge \C_{\opt}$, by \cref{clm:wedoreal} all left-over work at time $\C_\opt$ comes from tasks with $\sigma_i + t_i \in [\lambda, 1.5\lambda]$. By \cref{clm:leftover-work-small}, the total amount of such work is at most $\lambda - R(\lambda)$. Then, by \cref{clm:learn-late}, we know $R(\lambda) > .5\lambda$. Together, this implies that the total amount of leftover work is at most $.5\lambda \le .5\C_\opt$.

\paragraph{Case 2:} $\lambda < 2 \C_{\opt}/3$.\\

First note that $\lambda \ge \C_{\opt}/2$ or else $\tau_{i_\star}$ would be started on a slow machine at time $\C_\opt$. So, all left-over work at time $\C_{\opt}$ comes either from tasks with $\sigma_i + t_i \in [\lambda, 1.5\lambda]$, or from tasks with $\sigma_i + t_i \in [1.5\lambda, 1.5^2 \lambda]$. By \cref{clm:leftover-work-small}, we can therefore bound the total amount of leftover work by $(\lambda - R(\lambda))+ (1.5\lambda - R(1.5\lambda))$. Now, by \cref{clm:learn-late}, this quantity can be at most $(\lambda - .5\lambda) + (1.5\lambda - 1.25\lambda) = .75\lambda$. Since $\lambda < 2\C_\opt/3$, this is at most $.5\C_\opt$.
\end{proof}

\begin{remark}
In \cref{appendix} we show that \cref{alg:cancel} is optimal among
never-committing schedulers that never cancel implementations running on slow machines. This shows that improving on \cref{alg:cancel} will require a substantially different scheduler.
\end{remark}

\section{Extending Beyond the Massively Parallel Regime}
In \cref{thm:tgr}, we have shown that \cref{alg:2gr} is a $2$-competitive
instantly-committing scheduler in the Massively Parallel regime of the SPDP. In
this section, we will show that in fact, \cref{alg:2gr} is a $2$-competitive
scheduler even outside of the Massively Parallel regime, although the analysis
is slightly more complicated. This result is interesting in its own right,
resolving an open question from \cite{ku24}.However, we think that the main
virtue of this proof is that it serves as a proof-of-concept that results from
the conceptually simpler Massively Parallel regime can be adapted to apply to
the general SPDP: we conjecture that all upper bounds holding in the massively parallel regime should also hold in the general SPDP.\\ 

\cref{alg:2gr} is not a defined scheduler in the SPDP, because
we specify the decisions for which tasks to run, but do not specify how to
schedule the tasks. We extend $\tgr$ to the general SPDP as follows:
\begin{scheduler}\label{alg:2grSPDP} When task $\tau_i$ arrives:
\begin{itemize}
\item If $\sigma_i+t_i > 2\C^{t_i}$ parallelize $\tau_i$.
\item Otherwise, serialize $\tau_i$.
\end{itemize}
At every timestep, if there are $x$ serial jobs present, then $\tgr$ schedules the jobs by allocating a processor to each of the $\min(p,x)$ serial jobs with the most remaining work, and then allocating any remaining processors to an arbitrary parallel job (if a parallel job is present).\\
\end{scheduler}

Now we analyze $\tgr$. 
We say $\tgr$ is \defn{saturated} at time $t$ if $\tgr$ has no idle
processors at time $t$. 
\begin{lemma}\label{lem:unsatez}
If $\tgr$ is unsaturated right before finishing, then $\C_\tgr \le 2\C_\opt$.
\end{lemma}
\begin{proof}
We claim that if $\tgr$ is unsaturated at time $t$, then for each task $\tau_i$
present at time $t$, $\tau_i$ has been run on every time step since it arrived.
Suppose this is not the case. Then, there must have been some time step before
time $t$ when there were at least $p$ serial jobs with at least as much
remaining work as $\tau_i$. But then $\tau_i$ will finish at the same time as
these other jobs, contradicting the fact that $\tgr$ is unsaturated at time $t$.
Thus, if $\tgr$ is unsaturated at time $t$, then $t\le \sigma_i+t_i$ for some
$i$ such that $\tgr$ ran $\tau_i$ in serial. Thus, $t\le 2\C_\opt$, as desired.\\
\end{proof}

By virtue of \cref{lem:unsatez} it suffices to consider the case that $\tgr$
is saturated immediately before finishing.
Let $t_*$ be the final time in $[0,\C_\tgr)$ when $\tgr$ is unsaturated (we set
$t_* = 0$ if $\tgr$ is always saturated).
Let $i_*\in [n]$ be the smallest $i$ such that $t_i\ge t_*$; in fact we will
have $t_{i_*}=t_*$, since in order to transition from being
unsaturated to being saturated, some tasks must arrive.
For integer $i\in [i_*, n]$, let $\K_\alg^i$ denote the sum of $\pi_j$ for each
$\tau_j$ with $i_*\le j\le i$ that $\alg$ runs in parallel; If $\alg$
is an instantly-committing scheduler then $\K_\alg^{i}$ can be computed at time $t_i$.
Now we prove an analogue of \cref{lem:tgrinvar}.
\begin{lemma}\label{lem:invarREPRISE}
Fix a length $n$ TAP. 
For all $i\in [i_*, n]$, and for all instantly-committing
schedulers $\alg$, 
\begin{equation}
  \label{eq:invarREPRISE}
\K_\tgr^i \le (\C^i_\alg-t_*)p+\K^i_{\alg}.
\end{equation}
\end{lemma}
\begin{proof}
We prove \cref{eq:invarREPRISE} by induction on $i$.
The base case is $i=i_*$; here the claim trivially holds. 
Subsequently, note that if $\alg$ takes at least as much work as $\tgr$ on $\tau_i$ (i.e., either $\alg$ runs $\tau_i$ in parallel or $\tgr$ runs $\tau_i$ in serial) then $\K_\alg^{i}\ge \K_\tgr^{i}$, and
$\C_\alg^{i}\ge t_*$, so \cref{eq:invarREPRISE} is true.
In the case that $\alg$ serializes $\tau_i$ while $\tgr$ parallelizes
$\tau_i$ we now have:
\[ \C_\alg^{i} \ge \sigma_i+t_i \ge 2\C^{t_i}. \] 
One consequence of this is that $\opt(t_i)$ parallelizes $\tau_i$, and so we
have \cref{eq:invarREPRISE} for $(i,\opt(t_i))$. Thus, 
\[ \K_\tgr^{i} \le (\C^{t_i}-t_*)p+\K_{\opt(t_i)}^{i}\le (2\C^{t_i}-2t_*)p \le
(\C_\alg^{i}-t_*)p.\] 
\end{proof}

In the Massively Parallel regime \cref{thm:tgr} followed immediately from
\cref{lem:tgrinvar}. Slightly more work is required in the general setting, but
\cref{lem:invarREPRISE} is still very useful.
\begin{theorem}
$\tgr$ is a $2$-competitive instantly-committing scheduler in the SPDP.
\end{theorem}
\begin{proof}
Recall from \cref{lem:unsatez} that we need only consider the case that $\tgr$
ends saturated, and recall the definition of $t_*$.
For any scheduler $\alg$, let $\B_\alg$ denote the work that $\alg$ has left
immediately before time $t_*$, and let $\K_\alg$ be work that $\alg$ takes on tasks $\tau_i$ with
 $t_i\ge t_*$.
Because $\tgr$ ends saturated, we have 
\[\C_{\tgr} = t_* + (\K_\tgr + \B_\tgr)/p.\]
Applying \cref{lem:invarREPRISE} gives
\begin{equation}
  \label{eq:ctgrppp}
    t_* + (\K_\tgr + \B_\tgr)/p \le \C_\opt + (\K_\opt+\B_\tgr)/p .
\end{equation}
So, to conclude, it suffices to show that $\B_\tgr+\K_\opt\le p\C_\opt$. Let $S$ be the set of tasks that $\tgr$ has present immediately before time $t_*$. Let $W = \sum_{\tau_i \in S}\sigma_i$. Clearly $\B_\tgr \le W$. On the other hand, $\opt$ must take at least $W$ work on the tasks $S$, and can
have made at most $pt_*$ progress on these tasks by time $t_*$. Thus, 
\[ \B_\tgr \le W\le pt_* + \B_\opt. \] 
Therefore, 
\[ \B_\tgr+\K_\opt \le pt_* + \B_\opt+\K_\opt \le p\C_\opt. \] 
Using this in \cref{eq:ctgrppp} gives $\C_\tgr\le 2\C_\opt$.
\end{proof}

\printbibliography

@inproceedings{yao97,
  title={Probabilistic computations: Toward a unified measure of complexity},
  author={Yao, Andrew Chi-Chin},
  booktitle={18th Annual Symposium on Foundations of Computer Science (sfcs 1977)},
  pages={222--227},
  year={1977},
  organization={IEEE Computer Society}
}

@inproceedings{ku24,
  title={Scheduling Jobs with Work-Inefficient Parallel Solutions},
  author={Kuszmaul, William and Westover, Alek},
  booktitle={Proceedings of the 36th ACM Symposium on Parallelism in Algorithms and Architectures},
  pages={101--111},
  year={2024}
}

@inproceedings{Turek92,
  title={Approximate algorithms scheduling parallelizable tasks},
  author={Turek, John and Wolf, Joel L and Yu, Philip S},
  booktitle={Proceedings of the fourth annual ACM symposium on Parallel algorithms and architectures},
  pages={323--332},
  year={1992}
}

@inproceedings{Tiwari94,
  title={Scheduling malleable and nonmalleable parallel tasks},
  author={Ludwig, Walter and Tiwari, Prasoon},
  booktitle={Proceedings of the fifth annual ACM-SIAM symposium on Discrete algorithms},
  pages={167--176},
  year={1994}
}

@inproceedings{Mounie99,
	address = {New York, NY, USA},
	series = {{SPAA} '99},
	title = {Efficient approximation algorithms for scheduling malleable tasks},
	isbn = {978-1-58113-124-6},
	url = {https://doi.org/10.1145/305619.305622},
	doi = {10.1145/305619.305622},
	urldate = {2023-02-02},
	booktitle = {Proceedings of the eleventh annual {ACM} symposium on {Parallel} algorithms and architectures},
	publisher = {Association for Computing Machinery},
	author = {Mounie, Gregory and Rapine, Christophe and Trystram, Dennis},
	year = {1999},
	pages = {23--32},
}

@inproceedings{DuttonMao07,
	address = {USA},
	series = {{PDCS} '07},
	title = {Online scheduling of malleable parallel jobs},
	isbn = {978-0-88986-704-8},
	urldate = {2023-02-02},
	booktitle = {Proceedings of the 19th {IASTED} {International} {Conference} on {Parallel} and {Distributed} {Computing} and {Systems}},
	publisher = {ACTA Press},
	author = {Dutton, Richard A. and Mao, Weizhen},
	year = {2007},
	pages = {136--141},
}

@article{Graham69,
author = {Graham, R. L.},
title = {Bounds on Multiprocessing Timing Anomalies},
journal = {SIAM Journal on Applied Mathematics},
volume = {17},
number = {2},
pages = {416-429},
year = {1969},
doi = {10.1137/0117039},
URL = { https://doi.org/10.1137/0117039 },
eprint = { https://doi.org/10.1137/0117039 }
}

@article{YeChen18,
	title = {Online scheduling of moldable parallel tasks},
	volume = {21},
	url = {https://ideas.repec.org//a/spr/jsched/v21y2018i6d10.1007_s10951-018-0556-2.html},
	language = {en},
	number = {6},
	urldate = {2023-02-03},
	journal = {Journal of Scheduling},
	author = {Ye, Deshi and Chen, Danny Z. and Zhang, Guochuan},
	year = {2018},
	note = {Publisher: Springer},
	pages = {647--654},
}

@article{baker1983shelf,
  title={Shelf algorithms for two-dimensional packing problems},
  author={Baker, Brenda S and Schwarz, Jerald S},
  journal={SIAM Journal on Computing},
  volume={12},
  number={3},
  pages={508--525},
  year={1983},
  publisher={SIAM}
}

@article{GuoKang10,
	title = {Online scheduling of malleable parallel jobs with setup times on two identical machines},
	volume = {206},
	issn = {0377-2217},
	url = {https://www.sciencedirect.com/science/article/pii/S037722171000175X},
	doi = {10.1016/j.ejor.2010.03.005},
	language = {en},
	number = {3},
	urldate = {2023-02-09},
	journal = {European Journal of Operational Research},
	author = {Guo, Shouwei and Kang, Liying},
	month = nov,
	year = {2010},
	pages = {555--561},
}

@inproceedings{Turek94_moldable,
	address = {New York, NY, USA},
	series = {{SPAA} '94},
	title = {Scheduling parallelizable tasks to minimize average response time},
	isbn = {978-0-89791-671-4},
	url = {https://doi.org/10.1145/181014.181331},
	doi = {10.1145/181014.181331},
	urldate = {2023-02-09},
	booktitle = {Proceedings of the sixth annual {ACM} symposium on {Parallel} algorithms and architectures},
	publisher = {Association for Computing Machinery},
	author = {Turek, John and Ludwig, Walter and Wolf, Joel L. and Fleischer, Lisa and Tiwari, Prasoon and Glasgow, Jason and Schwiegelshohn, Uwe and Yu, Philip S.},
	year = {1994},
	pages = {200--209},
}

@inproceedings{Turek94_rigid,
  title={Scheduling parallel tasks to minimize average response time},
  author={Turek, John and Schwiegelshohn, Uwe and Wolf, Joel L and Yu, Philip S},
  booktitle={Proceedings of the fifth annual ACM-SIAM symposium on Discrete algorithms},
  pages={112--121},
  year={1994}
}

@inproceedings{HurnikPaulus08,
  title={Online algorithm for parallel job scheduling and strip packing},
  author={Hurink, Johann L and Paulus, Jacob Jan},
  booktitle={Approximation and Online Algorithms: 5th International Workshop, WAOA 2007, Eilat, Israel, October 11-12, 2007. Revised Papers 5},
  pages={67--74},
  year={2008},
  organization={Springer}
}

@article{Ye09,
  title={A note on online strip packing},
  author={Ye, Deshi and Han, Xin and Zhang, Guochuan},
  journal={Journal of Combinatorial Optimization},
  volume={17},
  number={4},
  pages={417--423},
  year={2009},
  publisher={Springer}
}

@article{horowitz76exact,
author = {Horowitz, Ellis and Sahni, Sartaj},
title = {Exact and Approximate Algorithms for Scheduling Nonidentical Processors},
year = {1976},
issue_date = {April 1976},
publisher = {Association for Computing Machinery},
address = {New York, NY, USA},
volume = {23},
number = {2},
issn = {0004-5411},
url = {https://doi.org/10.1145/321941.321951},
doi = {10.1145/321941.321951},
journal = {J. ACM},
month = {apr},
pages = {317–327},
numpages = {11}
}

@article{lenstra1990approximation,
  title={Approximation algorithms for scheduling unrelated parallel machines},
  author={Lenstra, Jan Karel and Shmoys, David B and Tardos, {\'E}va},
  journal={Mathematical programming},
  volume={46},
  pages={259--271},
  year={1990},
  publisher={Springer}
}

@article{ss-randomized-rounding,
author = {Schulz, Andreas S. and Skutella, Martin},
title = {Scheduling Unrelated Machines by Randomized Rounding},
journal = {SIAM Journal on Discrete Mathematics},
volume = {15},
number = {4},
pages = {450-469},
year = {2002},
doi = {10.1137/S0895480199357078},
URL = {https://doi.org/10.1137/S0895480199357078},
eprint = {https://doi.org/10.1137/S0895480199357078}
}

@article{two-types-unrelated,
author = {Vakhania, Nodari and Hernandez, Jose and Werner, Frank},
year = {2014},
month = {02},
pages = {1 - 9},
title = {Scheduling Unrelated Machines with Two Types of Jobs},
volume = {52},
journal = {International Journal of Production Research},
doi = {10.1080/00207543.2014.888789}
}

@inbook{molinaro19stochastic,
author = {Marco Molinaro},
title = {Stochastic $\ell_p$ Load Balancing and Moment Problems via the L-Function Method},
booktitle = {Proceedings of the 2019 Annual ACM-SIAM Symposium on Discrete Algorithms (SODA)},
chapter = {},
pages = {343-354},
doi = {10.1137/1.9781611975482.22},
URL = {https://epubs.siam.org/doi/abs/10.1137/1.9781611975482.22},
eprint = {https://epubs.siam.org/doi/pdf/10.1137/1.9781611975482.22}
}

@article{page20intersection,
title = {Makespan minimization on unrelated parallel machines with simple job-intersection structure and bounded job assignments},
journal = {Theoretical Computer Science},
volume = {809},
pages = {204-217},
year = {2020},
issn = {0304-3975},
doi = {https://doi.org/10.1016/j.tcs.2019.12.009},
url = {https://www.sciencedirect.com/science/article/pii/S0304397519307844},
author = {Daniel R. Page and Roberto Solis-Oba and Marten Maack},
}

@article{gupta2021stochastic,
  title={Stochastic load balancing on unrelated machines},
  author={Gupta, Anupam and Kumar, Amit and Nagarajan, Viswanath and Shen, Xiangkun},
  journal={Mathematics of Operations Research},
  volume={46},
  number={1},
  pages={115--133},
  year={2021},
  publisher={INFORMS}
}

@inproceedings{im2023improved,
  title={Improved approximations for unrelated machine scheduling},
  author={Im, Sungjin and Li, Shi},
  booktitle={Proceedings of the 2023 Annual ACM-SIAM Symposium on Discrete Algorithms (SODA)},
  pages={2917--2946},
  year={2023},
  organization={SIAM}
}

@inproceedings{deng2023generalized,
  title={Generalized unrelated machine scheduling problem},
  author={Deng, Shichuan and Li, Jian and Rabani, Yuval},
  booktitle={Proceedings of the 2023 Annual ACM-SIAM Symposium on Discrete Algorithms (SODA)},
  pages={2898--2916},
  year={2023},
  organization={SIAM}
}

@inproceedings{awerbuch1995load,
  title={Load balancing in the $L_p$ norm},
  author={Awerbuch, Baruch and Azar, Yossi and Grove, Edward F and Kao, Ming-Yang and Krishnan, P and Vitter, Jeffrey Scott},
  booktitle={Proceedings of IEEE 36th Annual Foundations of Computer Science},
  pages={383--391},
  year={1995},
  organization={IEEE}
}

@article{aspnes1997line,
  title={On-line routing of virtual circuits with applications to load balancing and machine scheduling},
  author={Aspnes, James and Azar, Yossi and Fiat, Amos and Plotkin, Serge and Waarts, Orli},
  journal={Journal of the ACM (JACM)},
  volume={44},
  number={3},
  pages={486--504},
  year={1997},
  publisher={ACM New York, NY, USA}
}

@inproceedings{caragiannis2008better,
  title={Better bounds for online load balancing on unrelated machines},
  author={Caragiannis, Ioannis},
  booktitle={Proceedings of the nineteenth annual ACM-SIAM symposium on Discrete algorithms},
  pages={972--981},
  year={2008}
}

@inproceedings{anand2011meeting,
  title={Meeting deadlines: How much speed suffices?},
  author={Anand, S and Garg, Naveen and Megow, Nicole},
  booktitle={Automata, Languages and Programming: 38th International Colloquium, ICALP 2011, Zurich, Switzerland, July 4-8, 2011, Proceedings, Part I 38},
  pages={232--243},
  year={2011},
  organization={Springer}
}

@inproceedings{gupta2017stochastic,
  title={Stochastic online scheduling on unrelated machines},
  author={Gupta, Varun and Moseley, Benjamin and Uetz, Marc and Xie, Qiaomin},
  booktitle={Integer Programming and Combinatorial Optimization: 19th International Conference, IPCO 2017, Waterloo, ON, Canada, June 26-28, 2017, Proceedings 19},
  pages={228--240},
  year={2017},
  organization={Springer}
}

@article{gupta2020greed,
  title={Greed works—online algorithms for unrelated machine stochastic scheduling},
  author={Gupta, Varun and Moseley, Benjamin and Uetz, Marc and Xie, Qiaomin},
  journal={Mathematics of operations research},
  volume={45},
  number={2},
  pages={497--516},
  year={2020},
  publisher={INFORMS}
}

@article{zhang2022randomized,
  title={Randomized selection algorithm for online stochastic unrelated machines scheduling},
  author={Zhang, Xiaoyan and Ma, Ran and Sun, Jian and Zhang, Zan-Bo},
  journal={Journal of Combinatorial Optimization},
  pages={1--16},
  year={2022},
  publisher={Springer}
}
\appendix


\section{Barriers Against Improved Schedulers}\label{appendix}
In this section we show that the schedulers of \cref{sec:eventual} and
\cref{sec:never} are optimal among natural restricted classes of schedulers. 
This highlights what changes must be made to the schedulers in order to have
hopes of achieving better competitive ratios.

First we show that among \defn{non-procrastinating} eventually-committing
schedulers (i.e., eventually-committing schedulers with the property that
whenever tasks are present, they will run at least one task), the scheduler
\cref{alg:eve} is optimal.
\begin{proposition}
Fix $\eps>0$.
Let $\xi \approx 1.677$ denote the real root of the polynomial $2x^3- 3x^2-1$. There is no
deterministic $(\xi-\eps)$-competitive non-procrastinating
eventually-committing scheduler.
\end{proposition}
\begin{proof}
It suffices to consider the case that $\eps< .001$.
Fix a non-procrastinating eventually-committing scheduler $\alg$.
Assume towards contradiction that $\alg$ is $(\xi-\eps)$-competitive.
We now describe a TAP $\mc{T}$ on which $\C_\alg \ge \xi\C_\opt$.
The TAP starts with $\tau_1=(1/\xi, 1/\xi^2,0)$.
Next, let $\tau_2=(1-\eps^2, 1/\xi-\eps^2, \eps^2)$.
Finally, at each time $t\in [\xi+1/\xi-2, 1-\eps^2] \cap (\N\eps^2)$, give a task
$\tau=(\infty,\eps^2,t)$. 

We now argue that $\alg$ must run all the tasks on the fast machine. Because
$\alg$ is a non-procrastinating $(\xi-\eps)$-competitive scheduler, $\alg$
must instantly start $\tau_1$ on the fast machine (in case there are no tasks
after $\tau_1$). Now we argue that $\alg$ runs $\tau_2$ on the fast machine as well.
Suppose that $\alg$ starts $\tau_2$ on a slow machine at some time $t$ with 
\begin{equation}
  \label{eq:noway}
1-\eps^2+t > (\xi-\eps)\C^{t}. 
\end{equation}
Then, $\alg$ would not be $(\xi-\eps)$-competitive on the truncated TAP
$\mc{T}^{t}$. Thus, $\alg$ must not start $\tau_2$ on a slow machine at any time
$t$ satisfying \cref{eq:noway}. We now show that \cref{eq:noway} holds for all
$t\ge \eps$, thus proving that $\alg$ must run $\tau_2$ on the fast machine.
For $t\in [\eps^2, \xi+1/\xi-2)$ 
we have $\C^{t}\le 1/\xi$, and $1-\eps^2+t\ge 1$, so \cref{eq:noway} holds.
For $t\ge \xi+1/\xi-2$ we have
\[\C^t \le \min(1, t-\xi+2+\eps^2).\]
Thus, it suffices to show:
\begin{equation}\label{eq:hardhard}
1-\eps^2+t > (\xi-\eps)\min(1, t-\xi+2+\eps^2).
\end{equation}
To show \cref{eq:hardhard} it suffices to check \cref{eq:hardhard} for
$t=\xi-1-\eps^2$ (by monotonicity of the inequality on either side of $t=\xi-1$).
At $t=\xi-1-\eps^2$ \cref{eq:hardhard} is:
\[\xi-2\eps^2 > \xi-\eps,\]
which  is true because $\eps< .001$.

We have now shown that $\alg$ runs all tasks in $\mc{T}$ on the fast machine.
Thus, (by definition of $\xi$)
\[ \C_\alg \ge 1/\xi^2 + 1/\xi-\eps^2 + 1-(1/\xi+\xi-2)-\eps^2 = \xi-2\eps^2. \] 
However,  $\C_\opt \le 1$.
This contradicts the assumption that $\alg$ is $(\xi-\eps)$-competitive.
\end{proof}

Now, we show that the $1.5$-competitive scheduler of \cref{sec:never} is optimal
among never-committing schedulers that don't cancel tasks on slow machines.
\begin{proposition}
Let $\alg$ be a deterministic never-committing scheduler that never cancels serial tasks.
Then, for any $\eps>0$, there is a TAP $\mc{T}$ with $n\le O(1)$ on which $\alg$
has is not $(1.5-\eps)$-competitive.
\end{proposition}
\begin{proof}
It suffices to consider the case that $\eps<.001$. The TAP is defined as
follows. First, $\tau_1=(2,1,0)$. Then, for each time $t\in [\eps^2,
1-\eps^2]\cap \N\eps^2$, a task $\tau=(\infty, \eps^2, t)$ arrives. We will show
that if $\alg$ starts $\tau_1$ on a slow machine at any time $t$ then
$\alg$ is not $(1.5-\eps)$-competitive on $\mc{T}^{t}$.
We show this by considering two cases.

\textbf{Case 1}: $\alg$ starts $\tau_1$ on a slow machine at time $t\in [0,1]$. \\
If $\alg$ does this, then $\C_\alg^{\mc{T}^t} \ge 2+t$. However, $\C^t \le t+1+\eps^2$.
Thus, 
\[\C_\alg^{\mc{T}^t}/\C^t \ge \frac{2+t}{t+1+\eps^2} \ge \frac{3}{2+\eps^2} > 1.5-\eps. \]
So $\alg$ cannot start $\tau_1$ on a slow machine at this time.

\textbf{Case 2}: $\alg$ starts $\tau_1$ on a slow machine at time $t\ge 1$.\\
If $\alg$ does this, then $\C_\alg \ge 2+t$. However, $\C_\opt\le 2$.
Thus, 
\[\C_\alg/\C_\opt \ge 1.5.\]

In conclusion, $\alg$ must run $\tau_1$ on the fast machine. But then 
\[\C_\alg\ge 3-\eps^2 > (1.5-\eps)\C_\opt = (1.5-\eps)2,\]
a contradiction.
\end{proof}

\section{Lower Bounds from \cite{ku24}}
In this section we state, for the reader's convenience, the lower bounds from \cite{ku24} against instantly- and eventually- committing schedulers.

\begin{proposition}[Kuszmaul, Westover \cite{ku24}]
Fix $\eps>0$. There is no deterministic $(2-\eps)$-competitive
instantly-committing scheduler.
\end{proposition}
\begin{proof}
Consider an $n$-task TAP where for each $i\in [n]$, the $i$-th task has $\sigma_i
= 2^{i}, \pi_i = 2^{i-1}$, and the arrival times are all very close to $0$. For each $i\in[n]$, it is possible to handle the first $i$ tasks in the TAP with
completion time $2^{i-1}$. Thus, a $(2-\eps)$-competitive scheduler cannot afford to run task $\tau_i$ on a slow machine. So, a $(2-\eps)$-competitive scheduler must run all tasks on the fast machine, giving completion time at least $2^{n}-1$ on this TAP, while $\C_\opt\le 2^{n-1}$. For large enough $n$ this implies that the scheduler is not actually $2-\eps$ competitive.
\end{proof}

\begin{proposition}[Kuszmaul, Westover \cite{ku24}]\label{lb:phi}
Fix $\eps>0$. There is no deterministic $(\phi-\eps)$-competitive eventually-committing scheduler,
where $\phi\approx 1.618$ is the golden ratio.
\end{proposition}
\begin{proof}
Suppose that $\alg$ is a $(\phi-\eps)$-competitive eventually-committing
scheduler. Let $\tau_1=(\phi,1,0)$; if there are no further tasks, $\alg$ must
run $\tau_1$ on the fast machine, starting at some time $t_0\le 1/\phi$. 
Let $\tau_2 = (\infty,\phi-t_0,t_0)$.
On this TAP, $\C_\opt = \phi$, while $\C_\alg\ge \phi + 1 = \phi^2$. So $\alg$ is not
$(\phi-\eps)$-competitive.
\end{proof}
\section{Randomized Lower Bounds}\label{sec:lower}
In this section we give lower bounds against randomized schedulers. 
Our main tool is Yao's minimax principle \cite{yao97}, which allows us to prove
a lower bound on the competitive ratio by exhibiting a distribution over TAPs,
and showing that any deterministic scheduler has poor expected cost on a random
TAP drawn from the distribution.

\begin{proposition}
For any $\eps>0.03$, there is no $(5/3-\eps)$-competitive DoA scheduler, even
with randomization.
\end{proposition}
\begin{proof}
Fix $N=25$.
For $k\in\N$, define $\mc{T}_k$ to be a length $k$ TAP with $\sigma_i=2^{i+1}, \pi_i=2^{i}$.
Let $\mathcal{D}$ denote the following distribution over TAPs: choose $k\in [N]$
uniformly randomly, and then output TAP $\mathcal{T}_k$.
By brute force enumeration of all possible deterministic
instantly-committing strategies, one can show that no such strategy is
$1.637$-competitive on this TAP.
\end{proof}

\begin{proposition}
For any $\eps>0$, there is no $((1+\sqrt{3})/4-\eps)$-competitive
eventually-committing scheduler, even with randomization.
\end{proposition}
\begin{proof}
In the proof of \cref{prop:canclb} we defined two TAPs, and showed that no
deterministic eventually-committing scheduler is
$((1+\sqrt{3})/2-\eps)$-competitive on both of the TAPs.
One can show that if we randomly choose between the two TAPs of
\cref{prop:canclb}, there is no deterministic eventually-committing scheduler
with expected competitive ratio $(1+\sqrt{3})/4-\eps$.
\end{proof}


\end{document}